\documentclass[10pt]{article}
\usepackage[pdftex]{graphicx}
\usepackage{caption}
\usepackage{subcaption}
\usepackage{dcolumn}
\usepackage{bm}
\usepackage{color}
\usepackage{amsfonts}
\usepackage{amsmath}
\usepackage{stmaryrd}
\usepackage{epsfig}
\usepackage[english]{babel}
\usepackage[all]{xy}
\usepackage{ulem}
\usepackage[top=25.4mm, bottom=25.4mm, left=31.7mm, right=32.2mm]{geometry}
\usepackage{amssymb}

\DeclareMathOperator{\arccosh}{arccosh}
\DeclareMathOperator{\arctanh}{arctanh}

\usepackage{xcolor}
\usepackage{multicol}
\usepackage{color}

\definecolor{LightBlue}{rgb}{0.8,0.8,1}

\renewcommand{\theequation}{\arabic{equation}}

\input{tcilatex}
\begin{document}

\title{\textbf{Lagrangian formulation of   the Darboux system}}
\author{ Lingling Xue$^{1}$, E.V. Ferapontov$^{2}$, M.V. Pavlov$^{1}$}
\date{}
\maketitle

\renewcommand{\baselinestretch}{1.25}

\vspace{-5mm}

\begin{center}
$^{1}$Department of Applied Mathematics\\[0pt]
Ningbo University\\[0pt]
Ningbo 315211, P.R. China \\[0pt]
\ \\[0pt]
$^{2}$Department of Mathematical Sciences \\[0pt]
Loughborough University \\[0pt]
Loughborough, Leicestershire LE11 3TU, UK \\[0pt]
\ \\[0pt]
e-mails: \\[1ex]
\texttt{xuelingling@nbu.edu.cn}\\[0pt]
\texttt{E.V.Ferapontov@lboro.ac.uk}\\[0pt]
\texttt{maksim@nbu.edu.cn}\\[0pt]
\end{center}

\bigskip

\begin{abstract}

The classical Darboux system governing rotation coefficients of three-dimensional metrics of diagonal curvature possesses an equivalent formulation as a sixth-order PDE for a scalar potential (related to the corresponding $\tau$-function). We demonstrate that this PDE is  Lagrangian and can be viewed as an explicit  scalar form of the `generating PDE of the KP hierarchy' as discussed recently in Nijhoff \cite{Nijhoff2024} in the Lagrangian multiform approach to the Darboux and KP hierarchies. Scalar Lagrangian formulations for differential-difference and fully discrete versions of the Darboux system are also constructed. In the first three cases (continuous and differential-difference with one and two discrete variables), the corresponding Lagrangians are expressible via elementary functions (logarithms), whereas the fully discrete case requires special functions (dilogarithms).

Remarkably, dispersionless limits of the above Lagrangians provide a complete list of 3D second-order integrable Lagrangians of the form $\int f(u_{xy}, u_{xt}, u_{yt})\, dxdydt$.

\end{abstract}

\bigskip


\bigskip


\newpage

\bigskip

\tableofcontents


\section{Introduction and summary of the main results}
\label{sec:intro}

Given a diagonal metric written in terms of the Lam\'e coefficients $H_i$,
$$
\sum_{i=1}^n H_i^2(dx^i)^2,
$$
let us introduce the rotation coefficients $\beta_{ki}$ via
\begin{equation}\label{c1Lame}
\partial _{k}H_{i}=\beta _{ki}H_{k},
\end{equation}
where $\partial_k$ denotes partial derivative with respect to  $x^k$. The requirement that the metric has `diagonal curvature' (that is, all curvature components  $R^i_{kkj}=0$ for  $i\ne j\ne k$), leads to the Darboux system
 for the rotation coefficients $\beta_{ki}$,
\begin{equation}\label{c1Dar}
\partial _{k}\beta _{ij}=\beta _{ik}\beta _{kj},
\end{equation}
no summation. This system has been extensively studied by Darboux in the context of $n$-orthogonal coordinate systems in $\mathbb{R}^n$ \cite{Darboux}. From the point of view of the modern theory of integrable systems, Darboux system constitutes three-dimensional $n$-wave system, for which linear system (\ref{c1Lame}) acts as the corresponding (adjoint) Lax representation, indeed, equations (\ref{c1Dar}) are the compatibility conditions of  (\ref{c1Lame}). Darboux system, as well as its differential-difference and fully discrete versions,
appear in a wide range of applications in differential geometry (both continuous and discrete),  in the context of KP hierarchy, in the theory of integrable systems of hydrodynamic type and Frobenius manifolds, etc., see e.g. \cite{BK, DS, KS, Dub, Schief03, Tsarev2, DS2000, Nijhoff2024} and references therein. It is well-known that under the so-called symmetric reduction, $\beta_{ij}=\beta_{ji}$, Darboux system (\ref{c1Dar}) can be written as a collection of compatible third-order PDEs for a single potential $u$, one PDE for every triple of distinct indices. This can be achieved by setting $\beta_{ij}=\sqrt {u_{ij}}$ (here and in what follows, lower indices of the potential $u$ indicate partial derivatives), leading to
$$
u_{ijk}=2\sqrt{u_{ij}u_{ik}u_{jk}}.
$$
It seems to be  less well-known, although explicitly mentioned in (Darboux   \cite{Darboux}, Chapter III, formula (13)) that for $n=3$, the full Darboux system (\ref{c1Dar}) can  be represented as a sixth-order  PDE for a single potential $u$ defined via the relations
\begin{equation}\label{c1u}
u_{ij}=\beta _{ij}\beta _{ji};
\end{equation}
note that relations (\ref{c1u}) are compatible modulo  Darboux system (\ref{c1Dar}). This potential was known to Lam\'e and Darboux, see e.g.   (\cite{Darboux}, Chapter III, formula (3)); it was observed later in \cite{DM, D} that $u$ is related to the $\tau$-function of  KP hierarchy via  $u=-\ln \tau$.

Our main observation is that the sixth-order PDE derived by Darboux is Lagrangian, thus, Darboux system (\ref{c1Dar}) implies a collection of compatible sixth-order Lagrangian PDEs for $u$, one PDE for every triple of distinct indices:
\begin{equation}\label{PDE}
\displaystyle \partial_i\partial_j\left(\frac{u_{ijk}+L}{2u_{ij}}\right)+\partial_i\partial_k\left(\frac{u_{ijk}+L}{2u_{ik}}\right)
+\partial_j\partial_k\left(\frac{u_{ijk}+L}{2u_{jk}}\right)-\partial_i\partial_j\partial_k\left(\ln ({u_{ijk}-L})\right)=0,
\end{equation}
where $L=\sqrt{u_{ijk}^2-4\, u_{ij}u_{ik}u_{jk}}$. Equation (\ref{PDE})  is represented in Euler-Lagrange form corresponding to a third-order Lagrangian,  $\int  F \, dx^i dx^j dx^k$, with the Lagrangian density
\begin{equation}\label{F}
F=L+u_{ijk}\ln (u_{ijk}-L);
\end{equation}
it is demonstrated in Section \ref{sec:comp} that all these Lagrangians can be combined into a Lagrangian multiform structure in the spirit of \cite{Nijhoff2023, Nijhoff2024}.
We show  that PDE (\ref{PDE}) is equivalent to  the `generating PDE  of the KP hierarchy' whose two-component form was proposed recently in \cite{Nijhoff2024}; see Section \ref{sec:Nij} for explicit  derivation of PDE (\ref{PDE}) from the formulae of  \cite{Nijhoff2024}. We point out that modulo  total derivatives, Lagrangian (\ref{F}) can be written in several equivalent forms such as
$$
F=L+\frac{1}{2}u_{ijk}\ln \frac{u_{ijk}-L}{u_{ijk}+L}\qquad {\rm or} \qquad F=L-u_{ijk}\arctanh \frac{L}{u_{ijk}}.
$$

\noindent{\bf Remark on multi-dimensional consistency.}
 Darboux system (\ref{c1Dar}), as well as its symmetric and generic forms  in terms of  potential $u$, are overdetermined systems, so the question of their consistency naturally arises. For the Darboux system (\ref{c1Dar}), this was thoroughly discussed in \cite{Nijhoff2023}, in both continuous and discrete settings. As the Darboux system is multi-dimensionally consistent and the potential $u$ is defined  via (\ref{c1u})  in a multi-dimensionally consistent way, all resulting PDEs in terms of $u$ should also be consistent. This can be demonstrated directly. Thus, for the symmetric reduction, $u_{ijk}=2\sqrt{u_{ij}u_{ik}u_{jk}}$, this follows from the formula for the fourth-order derivatives of $u$,
$$
u_{ijkl}=2(\sqrt{u_{ik}u_{jk}u_{il}u_{jl}}+\sqrt{u_{ij}u_{ik}u_{jl}u_{kl}}+\sqrt{u_{ij}u_{jk}u_{il}u_{kl}}),
$$
which are manifestly symmetric under the interchange of indices $i,j,k,l$. The general case $n\geq 4$ turns out to be more intricate: besides  sixth-order PDEs (\ref{PDE}) for every triple of distinct indices $i,j,k\in \{1, \dots, n\}$,  potential $u$ must also satisfy a collection of fourth-order and third-order PDEs, one fourth-order PDE and one third-order PDE for every quadruple of distinct indices $i,j,k,l\in \{1, \dots, n\}$. Altogether, they form an involutive system that follows, via multiform Euler-Lagrange equations, from a suitable Lagrangian 3-form. This is summarised in Section \ref{sec:comp}.

\bigskip
The structure of the paper is as follows.

In Section \ref{s:3+0} we discuss Lagrangian formulation of the Darboux system (\ref{c1Dar}) for $n=3$.

 In Sections \ref{s:2+1} and \ref{s:1+2} we construct analogous Lagrangian formulations for differential-difference Darboux systems with one and two discrete variables, respectively.
 The corresponding Lagrangian densities $F$ become considerably more complicated (due to natural asymmetry), although still expressible via elementary functions (logarithms).

 In Sections \ref{s:0+3} we do the same for the fully discrete  Darboux system:
 the corresponding Lagrangian density $F$ is expressible via special functions (dilogarithms).

Remarkably, dispersionless limits of the four Lagrangian densities $F$ of the Darboux system (including its differential-difference and fully discrete versions), provide a complete list of second-order integrable Lagrangians of the form
\begin{equation}\label{sec}
\int f(u_{xy},u_{xt},u_{yt})\ \text{d}x\text{d}y\text{d}t,
\end{equation}
as classified recently in \cite{XFP} (see Section \ref{sec:Lag} for a brief summary of these results). Up to certain natural equivalence, there exist exactly four integrable Lagrangians of type (\ref{sec}), the simplest of them being $f=\sqrt {u_{xy}u_{xt}u_{yt}}$. One can see that it can be obtained as dispersionless limit of Lagrangian (\ref{F}), by setting $u_{ijk}\to 0$.

In Section \ref{sec:Nij} we demonstrate that equation (\ref{PDE}) constitutes scalar form of the generating PDE of  KP hierarchy as discussed in \cite{Nijhoff2024}.

In Section \ref{sec:comp} we address  compatibility of the Darboux system written in terms of the potential $u$ for $n\geq 4$ and discuss its Lagrangian 3-form formulation.

\section{Scalar Lagrangian formulation of the Darboux system}
\label{sec:Dar}

In this section we  demonstrate that for $n=3$, the Darboux system can be written as a single sixth-order  Lagrangian equation in terms of the potential $u$ related to  the corresponding  $\tau$-function. We consider separately continuous, differential-difference and fully discrete cases.

\subsection{Continuous case}
\label{s:3+0}

Here the starting point is linear system (\ref{c1Lame}) for the Lam\'e coefficients $H_i$,
whose compatibility conditions constitute  Darboux system (\ref{c1Dar}) for the rotation coefficients $\beta_{ki}$.
Let us introduce a potential $u$ via  relations (\ref{c1u}),
$$
u_{ij}=\beta _{ij}\beta _{ji}.
$$
Our goal is to rewrite  Darboux system (\ref{c1Dar}) as a single PDE in terms of the potential $u$. Here we essentially follow \cite{Darboux}, Chapter III. Introducing the notation $m=\beta_{12}\beta_{23}\beta_{31}$ and $n=\beta_{13}\beta_{32}\beta_{21}$, one has $u_{123}=m+n$ and $mn=u_{12}u_{13}u_{23}$. Solving for $m$ and $n$ one obtains
\begin{equation}\label{c1pq}
m=\frac{u_{123}-L}{2}, \quad n=\frac{u_{123}+L}{2},
\end{equation}
where $L=\sqrt{u_{123}^{2}-4\, u_{12}u_{13}u_{23}}$. As the rotation coefficients $\beta_{ij}$ are assumed real, both $m$ and $n$ will be real as well. Thus, the expression under the square root cannot be negative. This property is automatically satisfied for a potential $u$ defined via (\ref{c1u}). Note that the choice of sign of the square root  will not affect the final formulae.

\begin{proposition} \label{pc1} Darboux system (\ref{c1Dar}) can be written as a single sixth-order PDE for $u$,
\begin{equation}\label{c1PDE}
\displaystyle \partial_1\partial_2\left(\frac{u_{123}+L}{2u_{12}}\right)+\partial_1\partial_3\left(\frac{u_{123}+L}{2u_{13}}\right)
+\partial_2\partial_3\left(\frac{u_{123}+L}{2u_{23}}\right)-\partial_1\partial_2\partial_3\left(\ln ({u_{123}-L})\right)=0.
\end{equation}
Equation (\ref{c1PDE}) is represented in Euler-Lagrange form corresponding to a Lagrangian  $ \int  F \, dx^1 dx^2 dx^3$, with the Lagrangian density
\begin{equation}\label{c1F}
F=L+u_{123}\ln (u_{123}-L).
\end{equation}
\end{proposition}

\noindent {\it Proof:} Parametrising relations (\ref{c1u}) in the form
$$
\begin{array}{c}
\beta_{12}=\sqrt{u_{12}}\, e^{\varphi}, \quad \beta_{21}=\sqrt{u_{12}}\, e^{-\varphi}, \\
\beta_{13}=\sqrt{u_{13}}\, e^{-\psi}, \quad \beta_{31}=\sqrt{u_{13}}\, e^{\psi}, \\
\beta_{23}=\sqrt{u_{23}}\, e^{\eta}, \quad \beta_{32}=\sqrt{u_{23}}\, e^{-\eta},
\end{array}
$$
and substituting into $m=\beta_{12}\beta_{23}\beta_{31}$, we obtain
\begin{equation}\label{c1short}
\varphi+\psi +\eta=\ln\frac{m}{\sqrt{u_{12}u_{13}u_{23}}}=\ln\frac{u_{123}-L}{2\sqrt{u_{12}u_{13}u_{23}}}=\frac{1}{2}\ln\frac{u_{123}-L}{u_{123}+L},
\end{equation}
where the second equality follows from (\ref{c1pq}). Under the same parametrisation, Darboux system  (\ref{c1Dar}) simplifies to
$$
\partial_3\varphi=\frac{L}{2u_{12}}, \quad \partial_2\psi=\frac{L}{2u_{13}}, \quad \partial_1\eta=\frac{L}{2u_{23}}.
$$
Applying to (\ref{c1short}) the operator $\partial_1\partial_2\partial_3$, one obtains a sixth-order PDE for $u$,
\begin{equation}\label{c1pde}
\partial_1\partial_2\left(\frac{L}{2u_{12}}\right)+\partial_1\partial_3\left(\frac{L}{2u_{13}}\right)+\partial_2\partial_3\left(\frac{L}{2u_{23}}\right)-\partial_1\partial_2\partial_3\left(\frac{1}{2}\ln \frac{u_{123}-L}{u_{123}+L}\right)=0,
\end{equation}
which is equivalent to (\ref{c1PDE}).  We emphasize that this PDE stays the same if we change the sign of $L$. Remarkably, equation (\ref{c1pde}) is  in Euler-Lagrange form. To reconstruct the corresponding Lagrangian,  note that for a Lagrangian density  of the form $F=F(u_{12},\, u_{13},\, u_{23},\, u_{123})$, the corresponding Euler-Lagrange equation is
$$
\partial_1\partial_2\left(\frac{\partial F}{\partial u_{12}}\right)+\partial_1\partial_3\left(\frac{\partial F}{\partial u_{13}}\right)+\partial_2\partial_3\left(\frac{\partial F}{\partial u_{23}}\right)-\partial_1\partial_2\partial_3\left(\frac{\partial F}{\partial u_{123}}\right)=0.
$$
Comparison with (\ref{c1pde}) gives the expressions for all first-order derivatives of $F$ which, on integration, leads to the Lagrangian density
$$
F=L+\frac{1}{2}u_{123}\ln\frac{u_{123}-L}{u_{123}+L}.
$$
Modulo total derivatives, this  density coincides with (\ref{c1F}), so we keep the same notation. $\square$

\bigskip

{\bf Symmetric reduction} of the Darboux system (\ref{c1Dar}) is specified by the condition
\begin{equation*}
\beta _{13}\beta _{32}\beta _{21}=\beta _{12}\beta _{23}\beta _{31},
\end{equation*}
which is equivalent to $L=0$ (up to a reparametrisation, this is also equivalent to $\beta_{ij}=\beta_{ji}$). It can be written as a third-order PDE for $u$ (see e.g. \cite{Egorov1901}, eq. (50), p. 227),
\begin{equation}\label{c1E}
u_{123}^2=4\, {u_{12}u_{13}u_{23}}.
\end{equation}
Equation (\ref{c1E}) can be obtained from Darboux system (\ref{c1Dar}) by setting
$$
\beta_{ij}=\sqrt{u_{ij}}.
$$
Note that reduction $L=0$ is compatible with PDE (\ref{c1PDE}).

\medskip

{ \bf Dispersionless limit } of the Lagrangian density (\ref{c1F}), obtained by setting $u_{123}\to 0$ coincides, modulo simple rescaling, with the second-order Lagrangian density
\begin{equation} \label{c1f}
f=\sqrt{u_{12}u_{13}u_{23}}.
\end{equation}
To be precise, dispersionless limit is obtained by scaling the  variables as $\tilde x^i=\epsilon  x^i, \, \tilde u=\epsilon^2  u$ and passing to the limit $\epsilon \to 0$.  In this case, second-order derivatives remain unchanged, $u_{ij}=\tilde u_{\tilde i \tilde j}$, while third-order derivatives (responsible for `dispersion') acquire a factor of $\epsilon$,  $u_{ijk}=\epsilon \tilde u_{\tilde i \tilde j \tilde k}$, and therefore should be set equal to zero. At the level of the corresponding Euler-Lagrange equations, dispersionless limit amounts to keeping terms of the lowest order in $\epsilon$ (under the above scaling) and setting all terms with higher powers of $\epsilon$  equal to zero. We refer to \cite{Zakharov} for a general discussion of dispersionless limit in 3D.

\medskip
{\bf Remark 1.} It would be interesting to find a direct route from the Lagrangian density of Darboux system (\ref{c1Dar}) constructed in \cite{Nijhoff2023},
$$
\frac{1}{2}(\beta_{12}\partial_3\beta_{21}-\beta_{21}\partial_3\beta_{12})+\frac{1}{2}(\beta_{23}\partial_1\beta_{32}-\beta_{32}\partial_1\beta_{23})+\frac{1}{2}(\beta_{31}\partial_2\beta_{13}-\beta_{13}\partial_2\beta_{31})
+\beta_{12}\beta_{23}\beta_{31}-\beta_{13}\beta_{32}\beta_{21},
$$
to the Lagrangian density $F$ of Proposition \ref{pc1}.

\medskip
{\bf Remark 2.} The proof of Proposition \ref{pc1} contains a B\"acklund transformation between the Darboux system and its $u$-form. Here we restrict to the three-component case, $n=3$.
First of all, given $\beta_{ij}$ satisfying the Darboux system (\ref{c1Dar}), the corresponding potential $u$ satisfying sixth-order PDE (\ref{c1PDE}) is defined via relations (\ref{c1u}) up to a transformation of the form $u\to u+f_1(x^1)+f_2(x^2)+f_3(x^3)$, where $f_i(x^i)$ are arbitrary functions of the indicated arguments; note that this transformation is a symmetry of PDE (\ref{c1PDE}).
Conversely, given a potential $u$ that satisfies PDE (\ref{c1PDE}), the corresponding $\beta_{ij}$ can be recovered from the equations
$$
\begin{array}{c}
\beta_{12}=\sqrt{u_{12}}\, e^{\varphi}, \quad \beta_{21}=\sqrt{u_{12}}\, e^{-\varphi}, \\
\beta_{13}=\sqrt{u_{13}}\, e^{-\psi}, \quad \beta_{31}=\sqrt{u_{13}}\, e^{\psi}, \\
\beta_{23}=\sqrt{u_{23}}\, e^{\eta}, \quad \beta_{32}=\sqrt{u_{23}}\, e^{-\eta},
\end{array}
$$
where $\varphi, \psi, \eta$
satisfy the relations
$$
\partial_3\varphi=\frac{L}{2u_{12}}, \quad \partial_2\psi=\frac{L}{2u_{13}}, \quad \partial_1\eta=\frac{L}{2u_{23}}, \quad \varphi+\psi +\eta=\ln\frac{u_{123}-L}{2\sqrt{u_{12}u_{13}u_{23}}},
$$
that are compatible modulo PDE (\ref{c1PDE}). These relations define $\varphi, \psi, \eta$  up to a transformation of the form $\varphi\to \varphi+g_1(x^1)-g_2(x^2), \ \psi\to \psi+g_3(x^3)-g_1(x^1), \ \eta\to \eta+g_2(x^2)-g_3(x^3)$, where $g_i(x^i)$ are arbitrary functions of the indicated arguments. The corresponding $\beta_{ij}$ will be defined up to a transformation of the form $\beta_{ij}\to \beta_{ij}e^{g_i(x^i)-g_j(x^j)}$, which is a symmetry of the Darboux system.
Thus, we have a genuine B\"acklund transformation that is multi-valued both ways.

\subsection{Differential-difference case (one discrete variable)}
\label{s:2+1}

Here the starting point is a differential-difference linear system
\begin{equation}\label{c2Lame}
\begin{array}{c}
\partial _{1}H_{2}=\beta _{12}H_{1},\quad  \partial _{1}H_{3}=\beta _{13}H_{1},\\
\partial _{2}H_{1}=\beta _{21}H_{2},\quad  \partial _{2}H_{3}=\beta _{23}H_{2},\\
\triangle_{3}H_{1}=\beta _{31}H_{3},\quad  \triangle _{3}H_{2}=\beta _{32}H_{3},\\
\end{array}
\end{equation}
where $\triangle_3=T_3-1$ is the discrete $x^3$-derivative and $T_3$ denotes  unit shift in the discrete variable $x^3$. The compatibility conditions lead to the differential-difference Darboux system,
\begin{equation}\label{c2Dar}
\begin{array}{c}
\partial _{1}\beta _{23}=\beta _{21}\beta _{13}, \quad  \partial _{1}\beta _{32}=\beta _{31}T_3\beta _{12}, \\
\partial _{2}\beta _{13}=\beta _{12}\beta _{23}, \quad  \partial _{2}\beta _{31}=\beta _{32}T_3\beta _{21}, \\
\triangle_{3}\beta _{12}=\beta _{13}\beta _{32}, \quad  \triangle _{3}\beta _{21}=\beta _{23}\beta _{31}.
\end{array}
\end{equation}
Let us introduce a potential $u$ via the relations
\begin{equation}\label{c2u}
u_{12}=\beta _{12}\beta _{21}, \quad \triangle_3u_{1}=\beta _{13}\beta _{31}, \quad \triangle_3u_{2}=\beta _{23}\beta _{32},
\end{equation}
which are compatible modulo the Darboux system (\ref{c2Dar}). Note that although we are using the same notation as in section \ref{s:3+0}, all variables have now different meaning; hope this will not cause any confusion as notation is restricted to the relevant section.
Introducing the notation $m=\beta_{12}\beta_{23}\beta_{31}$ and $n=\beta_{13}\beta_{32}\beta_{21}$, one has $\triangle_3u_{12}=m+n+\triangle_3u_1\triangle_3u_2$ and $mn=u_{12}\triangle_3u_{1}\triangle_3u_{2}$. Solving for $m$ and $n$ one obtains
\begin{equation}\label{c2pq}
m=\frac{\triangle_3u_{12}-\triangle_3u_1\triangle_3u_2-L}{2}, \quad n=\frac{\triangle_3u_{12}-\triangle_3u_1\triangle_3u_2+L}{2},
\end{equation}
where $L=\sqrt{(\triangle_3u_{12}-\triangle_3u_1\triangle_3u_2)^{2}-4\, u_{12}\triangle_3u_{1}\triangle_3u_{2}}$.

\begin{proposition} \label{pc2} Darboux system (\ref{c2Dar}) can be written as a single differential-difference equation for the potential $u$,
\begin{equation}\label{c2PDE}
\begin{array}{c}
\partial_1\partial_2\left(-\ln(1+\frac{m}{u_{12}})\right)+\partial_1\triangle_3\left(\frac{L+\triangle_3u_{12}}{2\triangle_3u_{1}}\right)
+\partial_2\triangle_3\left(\frac{L+\triangle_3u_{12}}{2\triangle_3u_{2}}\right)
-\partial_1\partial_2\triangle_3\left(\ln\frac{m}{u_{12}}\right)=0.
\end{array}
\end{equation}
Equation (\ref{c2PDE}) is represented in Euler-Lagrange form corresponding to a Lagrangian  $ \int F \, dx^1dx^2 \delta x^3$, with the Lagrangian density
\begin{equation}\label{c2F}
F=\frac{1}{2}L-u_{12}\ln \left(1+\frac{m}{u_{12}}\right)-\triangle_3u_{12}\ln \left(1+\frac{u_{12}}{m}\right).
\end{equation}
Here integration over the discrete variable $x^3$, denoted $\int \delta x^3$, is understood as summation over all $x^3$-translates of the density $F$.
\end{proposition}

\noindent {\it Proof:} Parametrising relations (\ref{c2u}) in the form
$$
\begin{array}{c}
\beta_{12}=\sqrt{u_{12}}\, e^{\varphi}, \quad \beta_{21}=\sqrt{u_{12}}\, e^{-\varphi}, \\
\beta_{13}=\sqrt{\triangle_3u_{1}}\, e^{-\psi}, \quad \beta_{31}=\sqrt{\triangle_3u_{1}}\, e^{\psi}, \\
\beta_{23}=\sqrt{\triangle_3u_{2}}\, e^{\eta}, \quad \beta_{32}=\sqrt{\triangle_3u_{2}}\, e^{-\eta},
\end{array}
$$
and substituting into the expression (\ref{c2pq}) for $m$, we obtain
\begin{equation}\label{c2short}
\varphi+\psi +\eta=\ln\frac{m}{\sqrt{u_{12}\triangle_3u_{1}\triangle_3u_{2}}}.
\end{equation}
Under the same parametrisation, Darboux system  (\ref{c2Dar}) takes the form
$$
\triangle_3\varphi=\frac{1}{2}\ln(1+\frac{\triangle_3u_{12}}{u_{12}})-\ln(1+\frac{m}{u_{12}}), \quad \partial_2\psi=\frac{L}{2\triangle_3u_{1}}+\frac{1}{2}\triangle_3u_2, \quad \partial_1\eta=\frac{L}{2\triangle_3u_{2}}-\frac{1}{2}\triangle_3u_1.
$$
Applying to (\ref{c2short}) the operator $\partial_1\partial_2\triangle_3$, one obtains a differential-difference equation in terms of $u$,
\begin{equation*}
\begin{array}{c}
\partial_1\partial_2\left(\frac{1}{2}\ln(1+\frac{\triangle_3u_{12}}{u_{12}})-\ln(1+\frac{m}{u_{12}})\right)+\partial_1\triangle_3\left(\frac{L}{2\triangle_3u_{1}}+\frac{1}{2}\triangle_3u_2\right)
+\partial_2\triangle_3\left(\frac{L}{2\triangle_3u_{2}}-\frac{1}{2}\triangle_3u_1\right)\\
-\partial_1\partial_2\triangle_3\left(\ln\frac{m}{\sqrt{u_{12}\triangle_3u_{1}\triangle_3u_{2}}}\right)=0.
\end{array}
\end{equation*}
On cancellation in the two middle terms, this equation takes the form
\begin{equation*}
\begin{array}{c}
\partial_1\partial_2\left(\frac{1}{2}\ln(1+\frac{\triangle_3u_{12}}{u_{12}})-\ln(1+\frac{m}{u_{12}})\right)+\partial_1\triangle_3\left(\frac{L}{2\triangle_3u_{1}}\right)
+\partial_2\triangle_3\left(\frac{L}{2\triangle_3u_{2}}\right)\\
-\partial_1\partial_2\triangle_3\left(\ln\frac{m}{\sqrt{u_{12}\triangle_3u_{1}\triangle_3u_{2}}}\right)=0.
\end{array}
\end{equation*}
Although it is already in Euler-Lagrange form,  we can simplify it by rewriting  the last term as
$$
\begin{array}{c}
-\partial_1\partial_2\triangle_3\left(\ln\frac{m}{\sqrt{u_{12}\triangle_3u_{1}\triangle_3u_{2}}}\right)=-\partial_1\partial_2\triangle_3\left(\ln\frac{m}{u_{12}}+\frac{1}{2}\ln u_{12}-\frac{1}{2}\ln \triangle_3u_{1}-\frac{1}{2}\ln \triangle_3 u_{2}\right)\\
=-\partial_1\partial_2\triangle_3\left(\ln\frac{m}{u_{12}}\right)-\partial_1\partial_2\left(\frac{1}{2}\ln(1+\frac{\triangle_3u_{12}}{u_{12}})\right)+\partial_1\triangle_3\left(\frac{\triangle_3u_{12}}{2\triangle_3u_1}\right)+\partial_2\triangle_3\left(\frac{\triangle_3u_{12}}{2\triangle_3u_2}\right),
\end{array}
$$
which, on rearrangement, results in the equivalent Euler-Lagrange form (\ref{c2PDE}).
To reconstruct the corresponding Lagrangian,  note that for a Lagrangian density  of the form $F=F(u_{12},\, \triangle_3u_1,\, \triangle_3u_2,\, \triangle_3u_{12})$, the corresponding Euler-Lagrange equation is
$$
\partial_1\partial_2\left(\frac{\partial F}{\partial u_{12}}\right)+\partial_1\triangle_3\left(\frac{\partial F}{\partial (\triangle_3u_{1})}\right)+\partial_2\triangle_3\left(\frac{\partial F}{\partial (\triangle_3u_{2})}\right)-\partial_1\partial_2\triangle_3\left(\frac{\partial F}{\partial (\triangle_3u_{12})}-\frac{\partial F}{\partial u_{12}}\right)=0.
$$
Comparing this with (\ref{c2PDE}) gives the expressions for all first-order derivatives of $F$ which, on integration, leads to the Lagrangian density
$$
F=\frac{1}{2}L-u_{12}\ln \left(1+\frac{m}{u_{12}}\right)-\triangle_3u_{12}\ln \left(1+\frac{u_{12}}{m}\right)+\frac{1}{2}\triangle_3u_{12}.
$$
This is equivalent to (\ref{c2F}) (modulo unessential total derivative term that does not effect the Euler-Lagrange equation). $\square$

\bigskip

{ \bf Symmetric reduction} (differential-difference version with one discrete variable)  of the Darboux system (\ref{c2Dar}) is specified by the condition
\begin{equation*}
\beta _{13}\beta _{32}\beta _{21}=\beta _{12}\beta _{23}\beta _{31},
\end{equation*}
which is equivalent to $L=0$. It can be written as
\begin{equation}\label{c2E}
(\triangle_3u_{12}-\triangle_3u_1\triangle_3u_2)^2=4\, {u_{12}\triangle_3u_{1}\triangle_3u_{2}}.
\end{equation}
Equation (\ref{c2E}) can be obtained from Darboux system (\ref{c2Dar}) by setting
$$
\begin{array}{c}
\beta_{12}=\beta_{21}=\sqrt{u_{12}}, \\
\ \\
 \beta_{13}=\sqrt{\triangle_3u_{1}}\, e^{-\frac{\triangle_3u}{2}}, \quad \beta_{31}=\sqrt{\triangle_3u_{1}}\, e^{\frac{\triangle_3u}{2}}, \quad
\beta_{23}=\sqrt{\triangle_3u_{2}}\, e^{-\frac{\triangle_3u}{2}}, \quad \beta_{32}=\sqrt{\triangle_3u_{2}}\, e^{\frac{\triangle_3u}{2}}.
\end{array}
$$
Note that reduction $L=0$ is compatible with equation (\ref{c2PDE}).

\medskip

{ \bf Dispersionless limit} of the Lagrangian density (\ref{c2F}), obtained by setting $\triangle_3u_{12}\to 0$ and $\triangle_3\to \partial_3$, coincides with the second-order Lagrangian density
\begin{equation} \label{c2f}
f=\frac{1}{2}l-u_{12}\ln\left(1-\frac{u_{13}u_{23}+l}{2u_{12}} \right),
\end{equation}
where $l=\sqrt{u_{13}^2u_{23}^2-4u_{12}u_{13}u_{23}}$. Modulo total derivatives, this Lagrangian density is equivalent to
\begin{equation*}
f=\frac{1}{2}l-2\, u_{12}\arctanh \frac{l}{u_{13}u_{23}}=\frac{1}{2}u_{13}u_{23}\sqrt{1-\frac{4u_{12}}{u_{13}u_{23}}}-2\, u_{12}\arctanh \sqrt{1-\frac{4u_{12}}{u_{13}u_{23}}}.
\end{equation*}
To be precise, dispersionless limit is obtained by scaling the  variables as $\tilde x^i=\epsilon  x^i, \, \tilde u=\epsilon^2  u$ and passing to the limit $\epsilon \to 0$.  Under this rescaling, unit shift  $T_3$ in the discrete variable $x^3$ becomes $\epsilon$-shift  $T_{\tilde 3}$,
second-order derivatives (both discrete and continuous) remain unchanged,  while third-order derivatives acquire a factor of $\epsilon$ and therefore should be set equal to zero.

\subsection{Differential-difference case (two discrete variables)}
\label{s:1+2}

Here the starting point is a differential-difference linear system
\begin{equation}\label{c3Lame}
\begin{array}{c}
\partial _{1}H_{2}=\beta _{12}H_{1},\quad  \partial _{1}H_{3}=\beta _{13}H_{1},\\
\triangle_{2}H_{1}=\beta _{21}H_{2},\quad  \triangle _{2}H_{3}=\beta _{23}H_{2},\\
\triangle_{3}H_{1}=\beta _{31}H_{3},\quad  \triangle _{3}H_{2}=\beta _{32}H_{3},\\
\end{array}
\end{equation}
where $\triangle_2$ and $\triangle_3$ denote discrete derivatives in the variables $x^2$ and $x^3$, respectively. The compatibility conditions lead to the differential-difference Darboux system,
\begin{equation}\label{c3Dar}
\begin{array}{c}
\partial _{1}\beta _{23}=\beta _{21}T_2\beta _{13}, \quad  \partial _{1}\beta _{32}=\beta _{31}T_3\beta _{12}, \\
\triangle_{2}\beta _{13}=\beta _{12}\beta _{23}, \quad  \triangle_{2}\beta _{31}=\beta _{32}T_3\beta _{21}, \\
\triangle_{3}\beta _{12}=\beta _{13}\beta _{32}, \quad  \triangle _{3}\beta _{21}=\beta _{23}T_2\beta _{31}.
\end{array}
\end{equation}
One also has  the  relations
$$
\begin{array}{c}
T_2\beta_{13}=\beta_{12}\beta_{23}+\beta_{13}, \quad T_3\beta_{12}=\beta_{13}\beta_{32}+\beta_{12},\quad
T_2\beta_{31}=\frac{\beta_{32}\beta_{21}+\beta_{31}}{1-\beta_{23}\beta_{32}}, \quad T_3\beta_{21}=\frac{\beta_{23}\beta_{31}+\beta_{21}}{1-\beta_{23}\beta_{32}},
\end{array}
$$
that follow from the Darboux system (\ref{c3Dar}). These relations allow one to eliminate shifts of $\beta$'s in subsequent calculations.
Let us introduce a potential $u$ via the relations
\begin{equation}\label{c3u}
\triangle_2u_{1}=\beta _{12}\beta _{21}, \quad \triangle_3u_{1}=\beta _{13}\beta _{31}, \quad \triangle_2\triangle_3u=-\ln(1-\beta_{23}\beta_{32}),
\end{equation}
which are compatible modulo the Darboux system (\ref{c3Dar}). Introducing the notation $m=\beta_{12}\beta_{23}\beta_{31}$ and $n=\beta_{13}\beta_{32}\beta_{21}$, one has
$$
\begin{array}{c}
\triangle_2\triangle_3u_1=(m+n+\triangle_2u_1+\triangle_3u_1)e^{\triangle_2\triangle_3u}-\triangle_2u_1-\triangle_3u_1, \\
mn=\triangle_2u_1\triangle_3u_1(1-e^{-\triangle_2\triangle_3u}).
\end{array}
$$
 Solving for $m$ and $n$ one obtains
 \begin{equation}\label{c3pq}
 m=\frac{b-\sqrt{b^2-4c}}{2}, \quad n=\frac{b+\sqrt{b^2-4c}}{2},
 \end{equation}
where
 $$
 \begin{array}{c}
 b=(\triangle_2\triangle_3u_1+\triangle_2u_1+\triangle_3u_1)e^{-\triangle_2\triangle_3u}-\triangle_2u_1-\triangle_3u_1, \\
 c=\triangle_2u_1\triangle_3u_1(1-e^{-\triangle_2\triangle_3u}).
 \end{array}
 $$

\begin{proposition} \label{pc3} Darboux system (\ref{c3Dar}) can be written as a single differential-difference equation for $u$,
\begin{equation}\label{c3PDE}
\begin{array}{c}
\partial_1\triangle_2\left( -\ln (1+\frac{m}{\triangle_2u_1})-\frac{1}{2}\triangle_2\triangle_3u\right)+\partial_1\triangle_3\left(-\ln (1+\frac{m}{\triangle_3u_1})-\frac{1}{2}\triangle_2\triangle_3u \right)\\
+\triangle_2\triangle_3\left( -\frac{\triangle_2\triangle_3u_1}{e^{\triangle_2\triangle_3u}-1}+\frac{\triangle_2u_1\triangle_3u_1}{m}+\frac{1}{2}\triangle_2u_1+\frac{1}{2}\triangle_3u_1 \right)\\
-\partial_1\triangle_2\triangle_3\left( \ln\frac{m}{\triangle_2u_{1}\triangle_3u_{1}(1-e^{-\triangle_2\triangle_3u})}  \right)=0.
\end{array}
\end{equation}
Equation (\ref{c3PDE}) is represented in Euler-Lagrange form corresponding to a Lagrangian $ \int F\, dx^1\delta x^2\delta x^3$, with the Lagrangian density
\begin{equation}\label{c3F}
\begin{array}{c}
F=\triangle_2 u_1 \ln \left(1+\frac{m}{\triangle_2 u_1}\right)
+\triangle_3 u_1 \ln \left(1+\frac{m}{\triangle_3 u_1}\right)
+\frac{1}{2} (\triangle_2 u_1 +\triangle_3 u_1 )\triangle_2 \triangle_3 u\\
+\triangle_2 \triangle_3 u_1 \ln \left(\triangle_2 \triangle_3 u_1-m\right).
\end{array}
\end{equation}
Here integration over the discrete variables $x^2, x^3$, denoted $ \int \delta x^2 \delta x^3$, is understood as summation over all $x^2, x^3$-translates of the density $F$.
\end{proposition}

\noindent {\it Proof:} Parametrising relations (\ref{c3u}) in the form
$$
\begin{array}{c}
\beta_{12}=\sqrt{\triangle_2u_1}\, e^{\varphi}, \quad \beta_{21}=\sqrt{\triangle_2u_1}\, e^{-\varphi}, \\
\beta_{13}=\sqrt{\triangle_3u_{1}}\, e^{-\psi}, \quad \beta_{31}=\sqrt{\triangle_3u_{1}}\, e^{\psi}, \\
\beta_{23}=\sqrt{1-e^{-\triangle_2\triangle_3u}}\, e^{\eta}, \quad \beta_{32}=\sqrt{1-e^{-\triangle_2\triangle_3u}}\, e^{-\eta},
\end{array}
$$
and substituting into the expressions for $m$, we obtain
\begin{equation}\label{c3short}
\varphi+\psi +\eta=\ln\frac{m}{\sqrt{\triangle_2u_{1}\triangle_3u_{1}(1-e^{-\triangle_2\triangle_3u})}}.
\end{equation}
Under the same parametrisation, Darboux system  (\ref{c3Dar}) takes the form
$$
\begin{array}{c}
\triangle_3\varphi=\ln \frac{\sqrt{\triangle_2u_1}\sqrt{\triangle_2\triangle_3u_1+\triangle_2u_1}}{m+\triangle_2u_1}-\triangle_2\triangle_3u, \quad
\triangle_2\psi=\ln \frac{\sqrt{\triangle_3u_1}\sqrt{\triangle_2\triangle_3u_1+\triangle_3u_1}}{m+\triangle_3u_1}, \\
\ \\
\partial_1\eta=-\frac{1}{2}\frac{\triangle_2\triangle_3u_1}{e^{\triangle_2\triangle_3u}-1}+\frac{\triangle_2u_1\triangle_3u_1}{m}+\triangle_2u_1.
\end{array}
$$
Applying to (\ref{c3short}) the operator $\partial_1\triangle_2\triangle_3$, one obtains a differential-difference equation in terms of $u$,
\begin{equation*}
\begin{array}{c}
\partial_1\triangle_2\left( \ln \frac{\sqrt{\triangle_2u_1}\sqrt{\triangle_2\triangle_3u_1+\triangle_2u_1}}{m+\triangle_2u_1}-\triangle_2\triangle_3u \right)+\partial_1\triangle_3\left(\ln \frac{\sqrt{\triangle_3u_1}\sqrt{\triangle_2\triangle_3u_1+\triangle_3u_1}}{m+\triangle_3u_1} \right)\\
+\triangle_2\triangle_3\left( -\frac{1}{2}\frac{\triangle_2\triangle_3u_1}{e^{\triangle_2\triangle_3u}-1}+\frac{\triangle_2u_1\triangle_3u_1}{m}+\triangle_2u_1  \right)
-\partial_1\triangle_2\triangle_3\left( \ln\frac{m}{\sqrt{\triangle_2u_{1}\triangle_3u_{1}(1-e^{-\triangle_2\triangle_3u})}}  \right)=0.
\end{array}
\end{equation*}
By rearranging linear terms, one can rewrite this equation in the equivalent (and more symmetric) Euler-Lagrange form,
\begin{equation}\label{c3pde}
\begin{array}{c}
\partial_1\triangle_2\left( \ln \frac{\sqrt{\triangle_2u_1}\sqrt{\triangle_2\triangle_3u_1+\triangle_2u_1}}{m+\triangle_2u_1}-\frac{1}{2}\triangle_2\triangle_3u\right)+\partial_1\triangle_3\left(\ln \frac{\sqrt{\triangle_3u_1}\sqrt{\triangle_2\triangle_3u_1+\triangle_3u_1}}{m+\triangle_3u_1}-\frac{1}{2}\triangle_2\triangle_3u \right)\\
+\triangle_2\triangle_3\left( -\frac{1}{2}\frac{\triangle_2\triangle_3u_1}{e^{\triangle_2\triangle_3u}-1}+\frac{\triangle_2u_1\triangle_3u_1}{m}+\frac{1}{2}\triangle_2u_1+\frac{1}{2}\triangle_3u_1 \right)\\
-\partial_1\triangle_2\triangle_3\left( \ln\frac{m}{\sqrt{\triangle_2u_{1}\triangle_3u_{1}(1-e^{-\triangle_2\triangle_3u})}}  \right)=0.
\end{array}
\end{equation}
This representation can be simplified if we rewrite  the last term as
$$
\begin{array}{c}
-\partial_1\triangle_2\triangle_3\left( \ln\frac{m}{\sqrt{\triangle_2u_{1}\triangle_3u_{1}(1-e^{-\triangle_2\triangle_3u})}}  \right) \\
=-\partial_1\triangle_2\triangle_3\left( \ln\frac{m}{\triangle_2u_{1}\triangle_3u_{1}(1-e^{-\triangle_2\triangle_3u})} +\ln \sqrt{\triangle_2u_1}+\ln \sqrt{ \triangle_3u_1}
+\ln \sqrt{1-e^{-\triangle_2\triangle_3u}}  \right) \\
=-\partial_1\triangle_2\triangle_3\left( \ln\frac{m}{\triangle_2u_{1}\triangle_3u_{1}(1-e^{-\triangle_2\triangle_3u})}\right)-\partial_1\triangle_2\left(\ln \frac{\sqrt{\triangle_2\triangle_3u_1+\triangle_2u_1}}{\sqrt{\triangle_2u_1}}\right)-\partial_1\triangle_3\left(\ln \frac{\sqrt{\triangle_2\triangle_3u_1+\triangle_3u_1}}{\sqrt{\triangle_3u_1}}\right)\\
-\frac{1}{2}\triangle_2\triangle_3\left(\frac{\triangle_2\triangle_3u_1}{e^{\triangle_2\triangle_3u}-1}\right),
\end{array}
$$
which, on rearrangement, results in the equivalent Euler-Lagrange form (\ref{c3PDE}).
To reconstruct the corresponding Lagrangian, note  that for a Lagrangian density   $F=F(\triangle_2u_1,\ \triangle_3u_1,\ \triangle_2\triangle_3u,\ \triangle_2\triangle_3u_{1})$, the corresponding  Euler-Lagrange equation is
$$
\begin{array}{c}
\partial_1\triangle_2\left(\frac{\partial F}{\partial (\triangle_2u_{1})}\right)+\partial_1\triangle_3\left(\frac{\partial F}{\partial (\triangle_3u_{1})}\right)+\triangle_2\triangle_3\left(\frac{\partial F}{\partial (\triangle_2\triangle_3 u)}\right)\\
-\partial_1\triangle_2\triangle_3\left(\frac{\partial F}{\partial (\triangle_2\triangle_3u_{1})}-\frac{\partial F}{\partial (\triangle_2u_{1})}-\frac{\partial F}{\partial (\triangle_3u_{1})}\right)=0.
\end{array}
$$
Comparing this with (\ref{c3PDE}) gives the expressions for all first-order derivatives of $F$ which, on integration, leads to the Lagrangian density
(\ref{c3F}) (modulo unessential  factor and total derivatives that do not effect the Euler-Lagrange equation).  $\square$

\bigskip

{ \bf Symmetric reduction} (differential-difference version with two discrete variables)  of the Darboux system (\ref{c3Dar}) is specified by the condition
\begin{equation*}
\beta _{13}\beta _{32}\beta _{21}=\beta _{12}\beta _{23}\beta _{31},
\end{equation*}
which can be written as
\begin{equation}\label{c3E}
(\triangle_2\triangle_3u_1+\triangle_2u_1+\triangle_3u_1)e^{-\triangle_2\triangle_3u}-\triangle_2u_1-\triangle_3u_1
=2\sqrt{\triangle_2u_1\triangle_3u_1(1-e^{-\triangle_2\triangle_3u})}.
\end{equation}
Equation (\ref{c3E}) can be obtained from Darboux system (\ref{c3Dar}) by setting
$$
\begin{array}{c}
\beta_{12}=\sqrt{\triangle_2u_1}\, e^{-\frac{\triangle_2u}{2}}, \quad \beta_{21}=\sqrt{\triangle_2u_1}\, e^{\frac{\triangle_2u}{2}}, \\
\ \\
 \beta_{13}=\sqrt{\triangle_3u_{1}}\, e^{-\frac{\triangle_3u}{2}}, \quad \beta_{31}=\sqrt{\triangle_3u_{1}}\, e^{\frac{\triangle_3u}{2}}, \\
 \ \\
\beta_{23}=\sqrt{1-e^{-\triangle_2\triangle_3u}}\, e^{\frac{\triangle_2u-\triangle_3u}{2}}, \quad \beta_{32}=\sqrt{1-e^{-\triangle_2\triangle_3u}}\, e^{\frac{\triangle_3u-\triangle_2u}{2}}.
\end{array}
$$

\medskip

{ \bf Dispersionless limit} of the Lagrangian density (\ref{c3F}), obtained by setting $\triangle_2\triangle_3u_{1}\to 0$ and $\triangle_2\to \partial_2, \ \triangle_3\to \partial_3$, coincides with the second-order Lagrangian density
\begin{equation}\label{c3f}
\begin{array}{c}
f=
u_{12} \ln \left(1+\frac{2 u_{13} }{L-u_{12}-u_{13}  }\right)
+ u_{13} \ln \left(1+\frac{2 u_{12} }{L-u_{12}-u_{13}  }\right)
+\frac{1}{2} (u_{12} + u_{13}  )u_{23},
\end{array}
\end{equation}
where $L=\sqrt{u_{12}^{2}+u_{13}^{2}- 2\, u_{12}u_{13}\coth \frac{u_{23}}{2}}$.
Using $u_{23}=\ln\frac{(L+u_{12}+u_{13})(L-u_{12}-u_{13})}{(L+u_{12}-u_{13})(L-u_{12}+u_{13})}$,
this Lagrangian density is equivalent to
\begin{equation*}
\begin{array}{c}
f= (u_{13}+u_{12}) {\rm \ arccoth \ }\frac{L}{u_{13}+u_{12}}-(u_{13}-u_{12}) {\rm \ arccoth \ } \frac{L}{u_{13}-u_{12}}.
\end{array}
\end{equation*}

\subsection{Discrete case}
\label{s:0+3}

We begin with a discrete linear system
\begin{equation}\label{c4Lame}
\begin{array}{c}
\triangle_{i}H_{j}=\beta _{ij}H_{i}, \end{array}
\end{equation}
whose compatibility conditions lead to the discrete Darboux system,
\begin{equation}\label{c4Dar}
\begin{array}{c}
\triangle _{i}\beta _{jk}=\beta _{ji}T_j\beta _{ik},
\end{array}
\end{equation}
$i\ne j\ne k \in \{1,2,3\}$. Here $\triangle_i=T_i-1$ is the discrete $x^i$-derivative and $T_i$ denotes  unit shift in discrete variable $x^i$.
One also has the relations
$$
T_i\beta_{jk}=\frac{\beta_{ji}\beta_{ik}+\beta_{jk}}{1-\beta_{ij}\beta_{ji}}
$$
that follow from the Darboux system (\ref{c4Dar}). These relations allow one to eliminate shifts of $\beta$'s in subsequent calculations.
In the discrete case,  potential $u$ is defined via the relations
\begin{equation}\label{c4u}
\triangle_1\triangle_2u=-\ln(1-\beta _{12}\beta _{21}), \quad \triangle_1\triangle_3u=-\ln(1-\beta _{13}\beta _{31}), \quad \triangle_2\triangle_3u=-\ln(1-\beta _{23}\beta _{32}),
\end{equation}
which are compatible modulo  Darboux system (\ref{c4Dar}); see (\cite{D2021}, formula (2.5) where  $u=-\ln \tau$).
Introducing the notation $m=\beta_{12}\beta_{23}\beta_{31}$ and $n=\beta_{13}\beta_{32}\beta_{21}$, one has
$$
\begin{array}{c}
\triangle_1\triangle_2\triangle_3u=-\ln(e^{-\triangle_1\triangle_2u}+e^{-\triangle_1\triangle_3u}+e^{-\triangle_2\triangle_3u}-m-n-2)-\triangle_1\triangle_2u-\triangle_1\triangle_3u-\triangle_2\triangle_3u,\\
mn=(1-e^{-\triangle_1\triangle_2u})(1-e^{-\triangle_1\triangle_3u})(1-e^{-\triangle_2\triangle_3u}).
\end{array}
$$
 Solving for $m$ and $n$ one obtains
 $$
 m=\frac{b-\sqrt{b^2-4c}}{2}, \quad n=\frac{b+\sqrt{b^2-4c}}{2},
 $$
 where
 $$
 \begin{array}{c}
 b=e^{-\triangle_1\triangle_2u}+e^{-\triangle_1\triangle_3u}+e^{-\triangle_2\triangle_3u}-2-e^{-\triangle_1 \triangle_2\triangle_3u-\triangle_1\triangle_2u-\triangle_1\triangle_3u-\triangle_2\triangle_3u},\\
 c=(1-e^{-\triangle_1\triangle_2u})(1-e^{-\triangle_1\triangle_3u})(1-e^{-\triangle_2\triangle_3u}).
 \end{array}
 $$

\begin{proposition} \label{pc4} Darboux system (\ref{c4Dar}) can be written as a single difference equation for $u$,
\begin{equation}\label{c4PDE}
\begin{array}{c}
\triangle_1\triangle_2\left(\ln \left(1+\frac{m}{1-e^{-\triangle_1\triangle_2u}}\right)+\frac{1}{2}\triangle_1\triangle_3u+\frac{1}{2}\triangle_2\triangle_3u\right) \\
+\triangle_1\triangle_3\left(\ln \left(1+\frac{m}{1-e^{-\triangle_1\triangle_3u}}\right)+\frac{1}{2}\triangle_1\triangle_2u+\frac{1}{2}\triangle_2\triangle_3u\right) \\
+\triangle_2 \triangle_3\left(\ln \left(1+\frac{m}{1-e^{-\triangle_2\triangle_3u}}\right)+\frac{1}{2}\triangle_1\triangle_2u+\frac{1}{2}\triangle_1\triangle_3u\right)\\
-\triangle_1\triangle_2\triangle_3\left(-\ln\frac{m}{{(1-e^{-\triangle_1\triangle_2u})(1-e^{-\triangle_1\triangle_3u})(1-e^{-\triangle_2\triangle_3u})}} \right)=0.
\end{array}
\end{equation}
Equation (\ref{c4PDE}) is represented in Euler-Lagrange form corresponding to a Lagrangian $ \int F\, \delta x^1\delta x^2\delta x^3$, with the Lagrangian density
\begin{equation}\label{c4F}
\begin{array}{c}
F=
-  \operatorname{Li}_2\left(e^{-\triangle_1\triangle_2u}\right)
- \operatorname{Li}_2\left(e^{-\triangle_1\triangle_3u}\right)
- \operatorname{Li}_2\left(e^{-\triangle_2\triangle_3u}\right)
- \operatorname{Li}_2\left(\frac{1}{1+m}\right)
\ \\
+\operatorname{Li}_2\left(\frac{e^{-\triangle_1\triangle_2u}}{1+m}\right)
+\operatorname{Li}_2\left(\frac{e^{-\triangle_1\triangle_3u}}{1+m}\right)
+\operatorname{Li}_2\left(\frac{e^{-\triangle_2\triangle_3u}}{1+m}\right)
+\operatorname{Li}_2\left(({1+m}){e^{-\triangle_1\triangle_2\triangle_3u}}\right)
\ \\
+(\triangle_1\triangle_2u+ \triangle_1\triangle_3u+ \triangle_2\triangle_3u+\ln(1+m))\ln(1+m)
\\
+\frac{1}{2}\triangle_1\triangle_2u \triangle_1\triangle_3u +\frac{1}{2}\triangle_1\triangle_2u\triangle_2\triangle_3u +\frac{1}{2}\triangle_1\triangle_3u\triangle_2\triangle_3u.
\end{array}
\end{equation}
Here integration over the discrete variables $x^1, x^2, x^3$, denoted $ \int \delta x^1 \delta x^2 \delta x^3$, is understood as summation over all $x^1, x^2, x^3$-translates of the density $F$, and
the dilogarithm function is defined as $\operatorname{Li}_2(x)=-\int_0^x \frac{\ln (1-t)}{t} dt$.
\end{proposition}
 \noindent {\it Proof:} Parametrising relations (\ref{c4u}) in the form
$$
\begin{array}{c}
\beta_{12}=\sqrt{1-e^{-\triangle_1\triangle_2u}}\, e^{\varphi}, \quad \beta_{21}=\sqrt{1-e^{-\triangle_1\triangle_2u}}\, e^{-\varphi}, \\
\beta_{13}=\sqrt{1-e^{-\triangle_1\triangle_3u}}\, e^{-\psi}, \quad \beta_{31}=\sqrt{1-e^{-\triangle_1\triangle_3u}}\, e^{\psi}, \\
\beta_{23}=\sqrt{1-e^{-\triangle_2\triangle_3u}}\, e^{\eta}, \quad \beta_{32}=\sqrt{1-e^{-\triangle_2\triangle_3u}}\, e^{-\eta},
\end{array}
$$
and substituting into the expression for $m$, we we obtain
\begin{equation}\label{c4short}
\begin{array}{c}
\varphi+\psi +\eta=\ln\frac{m}{\sqrt{(1-e^{-\triangle_1\triangle_2u})(1-e^{-\triangle_1\triangle_3u})(1-e^{-\triangle_2\triangle_3u})}}.
\end{array}
\end{equation}
Under the same parametrisation, Darboux system  (\ref{c4Dar}) gives
$$
\begin{array}{c}
\triangle_3\varphi=-\ln \left(1+\frac{m}{1-e^{-\triangle_1\triangle_2u}}\right)-\frac{1}{2}\ln(1-e^{-\triangle_1\triangle_2u})+\frac{1}{2}\ln(1-e^{-\triangle_1\triangle_2\triangle_3u-\triangle_1\triangle_2u})-\triangle_2\triangle_3u, \\
\triangle_2\psi=-\ln \left(1+\frac{m}{1-e^{-\triangle_1\triangle_3u}}\right)-\frac{1}{2}\ln(1-e^{-\triangle_1\triangle_3u})+\frac{1}{2}\ln(1-e^{-\triangle_1\triangle_2\triangle_3u-\triangle_1\triangle_3u})-\triangle_1\triangle_2u, \\
\triangle_1\eta=-\ln \left(1+\frac{m}{1-e^{-\triangle_2\triangle_3u}}\right)-\frac{1}{2}\ln(1-e^{-\triangle_2\triangle_3u})+\frac{1}{2}\ln(1-e^{-\triangle_1\triangle_2\triangle_3u-\triangle_2\triangle_3u})-\triangle_1\triangle_3u.
\end{array}
$$
Applying to relation (\ref{c4short}) the operator $\triangle_1\triangle_2\triangle_3$, one obtains a difference equation in terms of $u$,
 $$
\begin{array}{c}
\triangle_1\triangle_2\left(-\ln \left(1+\frac{m}{1-e^{-\triangle_1\triangle_2u}}\right)-\frac{1}{2}\ln(1-e^{-\triangle_1\triangle_2u})+\frac{1}{2}\ln(1-e^{-\triangle_1\triangle_2\triangle_3u-\triangle_1\triangle_2u})-\triangle_2\triangle_3u\right) \\
+\triangle_1\triangle_3\left(-\ln \left(1+\frac{m}{1-e^{-\triangle_1\triangle_3u}}\right)-\frac{1}{2}\ln(1-e^{-\triangle_1\triangle_3u})+\frac{1}{2}\ln(1-e^{-\triangle_1\triangle_2\triangle_3u-\triangle_1\triangle_3u})-\triangle_1\triangle_2u\right) \\
+\triangle_2 \triangle_3\left(-\ln \left(1+\frac{m}{1-e^{-\triangle_2\triangle_3u}}\right)-\frac{1}{2}\ln(1-e^{-\triangle_2\triangle_3u})+\frac{1}{2}\ln(1-e^{-\triangle_1\triangle_2\triangle_3u-\triangle_2\triangle_3u})-\triangle_1\triangle_3u\right)\\
-\triangle_1\triangle_2\triangle_3\left(\ln\frac{m}{\sqrt{(1-e^{-\triangle_1\triangle_2u})(1-e^{-\triangle_1\triangle_3u})(1-e^{-\triangle_2\triangle_3u})}} \right)=0.
\end{array}
$$
This equation can be simplified by eliminating square root in the last term, giving an equivalent equation
  $$
\begin{array}{c}
\triangle_1\triangle_2\left(-\ln \left(1+\frac{m}{1-e^{-\triangle_1\triangle_2u}}\right)-\triangle_2\triangle_3u\right) \\
+\triangle_1\triangle_3\left(-\ln \left(1+\frac{m}{1-e^{-\triangle_1\triangle_3u}}\right)-\triangle_1\triangle_2u\right) \\
+\triangle_2 \triangle_3\left(-\ln \left(1+\frac{m}{1-e^{-\triangle_2\triangle_3u}}\right)-\triangle_1\triangle_3u\right)\\
-\triangle_1\triangle_2\triangle_3\left(\ln\frac{m}{{(1-e^{-\triangle_1\triangle_2u})(1-e^{-\triangle_1\triangle_3u})(1-e^{-\triangle_2\triangle_3u})}} \right)=0.
\end{array}
$$
Redistributing linear terms in the first three summands gives an equivalent equation   (\ref{c4PDE}) which is already in Euler-Lagrange form.
 To reconstruct the corresponding Lagrangian, note that for a Lagrangian density   $F=F(\triangle_1\triangle_2u, \ \triangle_1\triangle_3u,\  \triangle_2\triangle_3u, \ \triangle_1\triangle_2\triangle_3u)$, the corresponding Euler-Lagrange equation is
$$
\begin{array}{c}
\triangle_1\triangle_2\left(\frac{\partial F}{\partial (\triangle_1\triangle_2u)}\right)+\triangle_1\triangle_3\left(\frac{\partial F}{\partial (\triangle_1 \triangle_3u)}\right)+\triangle_2\triangle_3\left(\frac{\partial F}{\partial (\triangle_2\triangle_3 u)}\right)\\
-\triangle_1\triangle_2\triangle_3\left(\frac{\partial F}{\partial (\triangle_1\triangle_2\triangle_3u)}-\frac{\partial F}{\partial (\triangle_1\triangle_2u)}-\frac{\partial F}{\partial (\triangle_1\triangle_3u)}-\frac{\partial F}{\partial (\triangle_2\triangle_3u)}\right)=0.
\end{array}
$$
Comparing this with (\ref{c4PDE}) gives the expressions for all first-order derivatives of $F$ which, on integration, leads to the Lagrangian density
(\ref{c4F}) (modulo unessential  total derivatives that do not effect the Euler-Lagrange equation).  $\square$

\bigskip

{ \bf Symmetric reduction} (discrete version)  of the Darboux system (\ref{c4Dar}) is specified by the condition
\begin{equation*}
\beta _{13}\beta _{32}\beta _{21}=\beta _{12}\beta _{23}\beta _{31},
\end{equation*}
which can be written as
\begin{equation}\label{c4E}
\begin{array}{c}
e^{-\triangle_1\triangle_2\triangle_3u-\triangle_1\triangle_2u-\triangle_1\triangle_3u-\triangle_2\triangle_3u}\\
=e^{-\triangle_1\triangle_2u}+e^{-\triangle_1\triangle_3u}+e^{-\triangle_2\triangle_3u}-2\sqrt{(1-e^{-\triangle_1\triangle_2u})(1-e^{-\triangle_1\triangle_3u})(1-e^{-\triangle_2\triangle_3u})}-2.
\end{array}
\end{equation}
Equivalently, it can be represented as an alternative form of the CKP equation (eq. (6.11) of \cite{Schief03}, $u=-\ln \tau$):
\begin{equation*}
\begin{array}{c}
(e^{-\triangle_1\triangle_2\triangle_3u-\triangle_1\triangle_2u-\triangle_1\triangle_3u-\triangle_2\triangle_3u}-e^{-\triangle_1\triangle_2u}-e^{-\triangle_1\triangle_3u}-e^{-\triangle_2\triangle_3u})^2\\
=4\, e^{-\triangle_1\triangle_2u-\triangle_1\triangle_3u-\triangle_2\triangle_3u}(e^{\triangle_1\triangle_2u}+e^{\triangle_1\triangle_3u}+e^{\triangle_2\triangle_3u}-e^{-\triangle_1\triangle_2\triangle_3u}-1).
\end{array}
\end{equation*}
Equation (\ref{c4E}) can be obtained from Darboux system (\ref{c4Dar}) by setting
$$
\beta_{jk}=\sqrt{1-e^{-\triangle_j\triangle_ku}}\, e^{\frac{\triangle_ju-\triangle_ku}{2}},
$$
see e.g. (\cite{DS2000}, formula (2.10),  $u=-\ln \tau$) for an equivalent parametrisation.

\medskip

{ \bf Dispersionless limit} of the Lagrangian density (\ref{c4F}), obtained by setting $\triangle_1\triangle_2\triangle_3 u\to 0$ and $\triangle_i\to \partial_i$, coincides with the second-order Lagrangian density
\begin{equation}\label{c4f}
\begin{array}{c}
f=-\operatorname{Li}_2(e^{-u_{12}})-\operatorname{Li}_2(e^{-u_{13}}) -\operatorname{Li}_2(e^{-u_{23}})- \operatorname{Li}_2\left(\frac{1}{1+m}\right)\\
\ \\
+\operatorname{Li}_2(\frac{e^{-u_{12}}}{1+m})+\operatorname{Li}_2(\frac{e^{-u_{13}}}{1+m})+\operatorname{Li}_2(\frac{e^{-u_{23}}}{1+m})
+\operatorname{Li}_2\left({1+m}\right)\\
\ \\
+(u_{12}+u_{13}+u_{23}+\ln(1+m))\ln(1+m)+\frac{1}{2}u_{12} u_{13} +\frac{1}{2}u_{12}u_{23} +\frac{1}{2}u_{13}u_{23},
\end{array}
\end{equation}
 where $ m=\frac{b-\sqrt{b^2-4c}}{2}$ and
 $$
 \begin{array}{c}
 b=e^{-u_{12}}+e^{-u_{13}}+e^{-u_{23}}-2-e^{-u_{12}-u_{13}-u_{23}},\quad
 c=(1-e^{-u_{12}})(1-e^{-u_{13}})(1-e^{-u_{23}}).
 \end{array}
 $$

\section{Second-order integrable Lagrangians in 3D}
\label{sec:Lag}

 In  paper \cite{XFP}, we have classified 3D second-order integrable Lagrangians of the form (\ref{sec}). The corresponding Euler-Lagrange equations are fourth-order quasilinear PDEs  for $u(x,y,t)$,
\begin{equation*}
\left( \frac{\partial f}{\partial u_{xy}}\right) _{xy}+\left( \frac{\partial
f}{\partial u_{xt}}\right) _{xt}+\left( \frac{\partial f}{\partial u_{yt}}
\right) _{yt}=0.
\end{equation*}
Integrability of  PDEs of this kind can be approached  based on the method of hydrodynamic reductions  \cite{FK04}. Without going into technical details (which can be found in \cite{FK04, XFP}), one calls a 3D equation `integrable'  if it can be decoupled in infinitely many ways into a pair of commuting 2D equations of hydrodynamic type. Exact solutions coming from such reductions can be considered as natural  analogues of multi-soliton/multi-gap solutions of conventional 3D soliton equations.
Crucially for  classification purposes, requiring the existence of such reductions  yields algorithmically computable differential constraints (integrability conditions) for the Lagrangian density $f$.
It turns out that modulo simple reparametrisations, there exist exactly four essentially different types of integrable Lagrangian densities $f$ of type (\ref{sec}).
  The first three of them are expressible via elementary functions:
 \medskip
$$
\begin{array}{c}
f= \sqrt{u_{xy}u_{xt}u_{yt}},\\
\ \\
f=u_{xt}u_{yt}\sqrt{1-\frac{2u_{xy}}{u_{xt}u_{yt}}}-2u_{xy}\arctanh \sqrt{1-\frac{2u_{xy}}{u_{xt}u_{yt}}},\\
\ \\
f= (u_{xt}+u_{xy}) {\rm \ arccoth \ }\frac{\sqrt{u_{xy}^{2}+u_{xt}^{2}- 2\, u_{xy}u_{xy}\coth \frac{u_{yt}}{2}}}{u_{xt}+u_{xy}}-(u_{xt}-u_{xy}) {\rm \ arccoth \ } \frac{\sqrt{u_{xy}^{2}+u_{xt}^{2}- 2\, u_{xy}u_{xy}\coth \frac{u_{yt}}{2}}}{u_{xt}-u_{xy}}.
\end{array}
$$
The fourth  density, expressible via dilogarithm function, is considerably more complicated:
$$
\begin{array}{c}
f=-\operatorname{Li}_2(e^{-u_{xy}})-\operatorname{Li}_2(e^{-u_{xt}}) -\operatorname{Li}_2(e^{-u_{yt}})- \operatorname{Li}_2\left(\frac{1}{1+m}\right)\\
\ \\
+\operatorname{Li}_2(\frac{e^{-u_{xy}}}{1+m})+\operatorname{Li}_2(\frac{e^{-u_{xt}}}{1+m})+\operatorname{Li}_2(\frac{e^{-u_{yt}}}{1+m})
+\operatorname{Li}_2\left({1+m}\right)\\
\ \\
+(u_{xy}+u_{xt}+u_{yt}+\ln(1+m))\ln(1+m)+\frac{1}{2}u_{xy} u_{xt} +\frac{1}{2}u_{xy}u_{yt} +\frac{1}{2}u_{xt}u_{yt},
\end{array}
$$
where $\operatorname{Li}_2$ is the dilogarithm function and $m$ is defined as in formula (\ref{c4f}). We refer to \cite{XFP} for an alternative geometric representation of the fourth density and its remarkable link  to spherical/hyperbolic trigonometry and Schl\"afly-type formulas.

Comparison with Sections \ref{s:3+0} -- \ref{s:0+3}  shows that, modulo simple rescalings,  the above Lagrangian densities are nothing but dispersionless limits of the Lagrangian densities governing Darboux hierarchy, a connection we have not anticipated when attempting the classification problem.

\section{Generating PDE of the KP hierarchy}
\label{sec:Nij}

In this section we provide details of  derivation of the sixth-order integrable Lagrangian PDE (\ref{c1PDE}) from  the two-component system (2.14)-(2.16) of Nijhoff \cite{Nijhoff2024}.
We follow the notation of \cite{Nijhoff2024}. Let $u$ and $v$ be functions of the three independent variables $\xi, \sigma, \tau$. Introduce the quantity
$$
{\digamma}:=\frac{4+\frac{u_{\sigma \tau}}{(1+u_{\sigma})(1+u_{\tau})}\partial_{\xi} \ln \left( \frac{u^2_{\sigma \tau}}{(1+u_{\sigma})(1+u_{\tau})} \right)}{2v+\partial_{\xi}\ln \frac{1+u_{\tau}}{1+u_{\sigma}}};
$$
(note that the above $u$ and ${\digamma}$ have different meaning than the analogous variables in our paper). Generating equations of the KP hierarchy is a system of PDEs for the two dependent variables $u$ and $v$ \cite{Nijhoff2024}:
\begin{align*}
v&=\partial_{\xi}\ln\left((1+u_{\sigma}) {\digamma}+\frac{u_{\sigma \tau}}{1+u_{\tau}}\right)+\frac{2}{{\digamma}+\frac{u_{\sigma \tau}}{(1+u_{\sigma})(1+u_{\tau})}}\\
&=-\partial_{\xi}\ln\left((1+u_{\tau}) {\digamma}-\frac{u_{\sigma \tau}}{1+u_{\sigma}}\right)+\frac{2}{{\digamma}-\frac{u_{\sigma \tau}}{(1+u_{\sigma})(1+u_{\tau})}},
\end{align*}
$$
2v_{\sigma \tau}=\partial_{\sigma}\left [(1+u_{\tau})v\left({\digamma}- \frac{u_{\sigma \tau}}{(1+u_{\sigma})(1+u_{\tau})} \right)\right ]-
\partial_{\tau}\left [(1+u_{\sigma})v\left({\digamma}+ \frac{u_{\sigma \tau}}{(1+u_{\sigma})(1+u_{\tau})} \right)\right].
$$
This purely continuous system for $u$ and $v$ was derived  by eliminating all the lattice shifts from a suitable collection of semi-discrete flows of the KP hierarchy. Here we quote \cite{Nijhoff2024}, p. 96: {\it `If one were to eliminate $v$ from this system, and
obtain a single PDE in terms of $u$, the highest derivative would be $u_{\xi \xi \sigma \sigma \tau \tau}$. We will refer
to the above coupled system of PDEs as the generating PDE of the KP hierarchy'}. Below, we show that the elimination procedure results in the sixth-order integrable Lagrangian PDE (\ref{c1PDE}). It goes as follows.
Introducing
$$
\alpha=\frac{u_{\sigma\tau}}{(1+u_\sigma)(1+u_\tau)},
$$
we can rewrite the above equations in the form
\begin{equation}\label{ff1}
 {\digamma}=\frac{4+\alpha \partial_\xi  \ln{(\alpha u_{\sigma \tau})}}{2v+ \partial_\xi \ln\left(\frac{1+u_\tau}{1+u_\sigma}\right)},
 \end{equation}
 \begin{equation}\label{ff2}
v=\partial_\xi  \ln\left[(1+u_\sigma)({\digamma}+\alpha)\right]+\frac{2}{{\digamma}+\alpha}, \quad v=-\partial_\xi  \ln\left[(1+u_\tau)({\digamma}-\alpha)\right]+\frac{2}{{\digamma}-\alpha},
 \end{equation}
 \begin{equation}\label{ff3}
2 v_{\sigma\tau}=\partial_\sigma\left[(1+u_\tau)v({\digamma}-\alpha)\right]-\partial_\tau\left[(1+u_\sigma)v({\digamma}+\alpha)\right].
\end{equation}
One can show that  these equations are not independent, in particular, equation (\ref{ff1}) is a corollary of  equations (\ref{ff2}), (\ref{ff3}). Adding equations \eqref{ff2} gives
\begin{equation}\label{ff0}
2v=-\partial_\xi \ln\frac{1+u_\tau}{1+u_\sigma}+\partial_\xi \ln\frac{{\it F}+\alpha}{{\digamma}-\alpha}+\frac{4{\digamma}}{{\digamma}^2-\alpha^2}.
\end{equation}
Subtracting equations \eqref{ff2} gives
\begin{equation}\label{ff00}
\partial_\xi \ln(1+u_\sigma)(1+u_\tau)+\partial_\xi \ln({\digamma}^2-\alpha^2)-\frac{4{\alpha}}{{\digamma}^2-\alpha^2}=0.
\end{equation}
Inserting \eqref{ff0} into \eqref{ff1} we have
\begin{equation}\label{LL}
{\digamma}\partial_\xi \ln\frac{{\digamma}+\alpha}{{\digamma}-\alpha}+\frac{4{\digamma}^2}{{\digamma}^2-\alpha^2}=4+\alpha \partial_\xi  \ln{(\alpha u_{\sigma \tau})}.
\end{equation}
Multiplying both sides of equation (\ref{LL}) my $\alpha$ and adding to it equation (\ref{ff00}) multiplied by ${\digamma}^2$ gives
$$
({\digamma}^2)_\xi-{\digamma}^2\partial_\xi \ln\left(\frac{\alpha}{u_{\sigma\tau}}\right)=\left(\frac{\alpha} {u_{\sigma\tau}}\right)\big[4u_{\sigma\tau}+ (\alpha u_{\sigma\tau})_\xi\big].
$$
Introducing a new variable $w$ such that $u=-w_{\xi}-\sigma-\tau$, we have $\alpha=\frac{-w_{\sigma \tau \xi}}{w_{\sigma \xi} w_{\tau \xi}}$ and the above equation integrates to
$$
{\digamma}^2=\alpha^2+ \frac{c-4w_{\sigma\tau}}{w_{\sigma\xi}w_{\tau\xi}}=\frac{w_{\sigma\tau\xi}^2+(c-4w_{\sigma\tau})w_{\sigma\xi}w_{\tau\xi}}{w_{\sigma\xi}^2w_{\tau\xi}^2},
$$
where $c$ is an integration constant which can be set equal to zero  (strictly speaking, $c$ can be a function of $\sigma$ and $\tau$, however, it can be absorbed into $w$). Thus, we have ${\digamma}=\frac{L}{w_{\sigma\xi}w_{\tau\xi}}$ where $L=\sqrt{w_{\sigma\tau\xi}^2-4 w_{\sigma\tau}w_{\sigma\xi}w_{\tau\xi}}$.  Finally, substituting \eqref{ff0} into the left-hand side of \eqref{ff3} and substituting \eqref{ff2} into the right-hand side of \eqref{ff3}, one has
$$
\left( \ln\frac{{\digamma}+\alpha}{{\digamma}-\alpha}\right)_{\xi\sigma\tau}+\left(\frac{4{\digamma}}{{\digamma}^2-\alpha^2}\right)_{\sigma\tau}-\left( \ln\frac{w_{\tau\xi}}{w_{\sigma\xi}}\right)_{\xi\sigma\tau}
-(w_{\tau\xi}({\digamma}-\alpha))_{\sigma\xi}-(w_{\sigma\xi}({\digamma}+\alpha))_{\tau\xi}=0,
$$
so that
$$
\left( \frac{L}{w_{\sigma\tau}}\right)_{\sigma\tau}+\left( \frac{L}{w_{\sigma\xi}}\right)_{\sigma\xi}
+\left( \frac{L}{w_{\tau\xi}}\right)_{\tau\xi}-\left( \ln\frac{L-w_{\sigma\tau\xi}}{L+w_{\sigma\tau\xi}}\right)_{\xi\sigma\tau}=0,
$$
which is equivalent to the  sixth-order integrable Lagrangian PDE (\ref{c1PDE}).

Given a solution $w$ of the above sixth-order PDE, the corresponding solutions $u, v$ of the Nijhoff system can be found (without  integration!) from the formulas
$u=-w_{\xi}-\sigma-\tau$ and ${\digamma}=\frac{L}{w_{\sigma\xi}w_{\tau\xi}}$ (note that $\digamma$ has $v$ in the denominator). Thus, the transformation $w\to \{u,v\}$ is a differential substitution.

\section{Compatibility of the Darboux system}
\label{sec:comp}

In this Section, we discuss the structure of the Darboux system written in terms of potential $u$.

For $n=3$, Darboux system (\ref{c1Dar}) is B\"acklund-related to a single sixth-order PDE (\ref{PDE}), so there are no compatibility issues.

The general case $n\geq 4$ is more interesting and non-trivial. It turns out that, besides the sixth-order PDEs (\ref{PDE}) for every triple of distinct indices $i,j,k\in \{1,\dots, n\}$, there exists a number of fourth-order and third-order PDEs for $u$, one fourth-order PDE and one third-order PDE for every quadruple of distinct indices $i,j,k,l\in \{1, \dots, n\}$. Altogether, they form an involutive system. These additional PDEs can be obtained as follows. Let us parametrise rotation coefficients in the form
\begin{equation}\label{5-beta}
\beta_{ij}=\sqrt{u_{ij}} e^{\varphi_{ij}}
\end{equation}
where $\varphi_{ij}$ is skew-symmetric, $\varphi_{ij}=-\varphi_{ji}$, and introduce the notation
$$
m_{ijk}=\beta_{ij}\beta_{jk}\beta_{ki}, \quad L_{ijk}=m_{ikj}-m_{ijk}.
$$
Note that $m_{ijk}$ have cyclic symmetry, whereas $L_{ijk}$ are totally skew-symmetric. Note also the easily verifiable relations
$$
m_{ijk}m_{ikj}=u_{ij}u_{ik}u_{jk}, \quad m_{ijk}+m_{ikj}=u_{ijk},
$$
which give
\begin{equation}\label{Lijk}
m_{ijk}=\frac{u_{ijk}-L_{ijk}}{2},\quad L_{ijk}^2=u_{ijk}^2-4u_{ij}u_{ik}u_{jk}.
\end{equation}
Inserting \eqref{5-beta} into
Darboux system (\ref{c1Dar}) gives
\begin{equation}\label{5-eq2}
T_{ijk}\equiv \varphi_{ij}+\varphi_{jk}+\varphi_{ki}=\frac{1}{2}\ln \frac{u_{ijk}-L_{ijk}}{u_{ijk}+L_{ijk}},\qquad       \partial_k\varphi_{ij}=\frac{L_{ijk}}{2u_{ij}}.
\end{equation}
Third-order PDEs for $u$ result from the identities $T_{ijk}-T_{ijl}+T_{ikl}-T_{jkl}=0$, which yield
\begin{equation*}
m_{ijk}m_{ikl}u_{jl}-m_{ijl}m_{jkl}u_{ik}=0,
\end{equation*}
equivalently,
\begin{equation}\label{5-c1}
A_{ijkl}\equiv \frac{(u_{ijk}-L_{ijk})(u_{ikl}-L_{ikl})u_{jl}}{(u_{ijl}-L_{ijl})(u_{jkl}-L_{jkl})u_{ik}}=1.
\end{equation}
By using  \eqref{Lijk}, one can show that for every quadruple of distinct indices $i,j,k,l$, all relations \eqref{5-c1} are equivalent to each other, and
are also equivalent to
 \begin{equation*}
 \begin{array}{ccc}
{\widetilde A}_{ijkl} \equiv u_{ijk}\left(\frac{L_{ijl}}{2u_{ij}}+\frac{L_{jkl}}{2u_{jk}}+\frac{L_{kil}}{2u_{ki}}\right)
-u_{ijl}\left(\frac{L_{ijk}}{2u_{ij}}+\frac{L_{jlk}}{2u_{jl}}+\frac{L_{lik}}{2u_{li}}\right)
\\
\qquad \qquad +u_{ikl}\left(\frac{L_{ikj}}{2u_{ik}}+\frac{L_{klj}}{2u_{kl}}+\frac{L_{lij}}{2u_{li}}\right)
-u_{jkl}\left(\frac{L_{kji}}{2u_{kj}}+\frac{L_{kli}}{2u_{kl}}+\frac{L_{lji}}{2u_{lj}}\right)=0.
\end{array}
\end{equation*}
Fourth-order PDEs for $u$ result from the requirement of compatibility of  equations \eqref{5-eq2}:
\begin{equation}\label{5-c2}
\begin{array}{ccc}
B_{ijlk}\equiv\partial_l\left(\frac{L_{ijk}}{2u_{ij}}\right)-\partial_k\left(\frac{L_{ijl}}{2u_{ij}}\right)=0,
\\
C_{ijkl}\equiv\frac{L_{ijl}}{2u_{ij}}+\frac{L_{jkl}}{2u_{jk}}+\frac{L_{kil}}{2u_{ki}}-\partial_l\left(\frac{1}{2}\ln \frac{u_{ijk}-L_{ijk}}{u_{ijk}+L_{ijk}}\right)=0.
\end{array}
\end{equation}
By using the second set of relations \eqref{Lijk} and third-order relations \eqref{5-c1}, all fourth-order relations \eqref{5-c2} can be shown to be equivalent to each other.
Thus, there is in fact only one third-order PDE and one fourth-order PDE for $u$ for every quadruple of distinct indices $i,j,k,l$.
Sixth-order PDEs for $u$ take the form
\begin{equation}\label{5-c3}
E_{ijk}\equiv \partial_i\partial_j\left(\frac{L_{ijk}}{2u_{ij}}\right)
+\partial_j\partial_k\left(\frac{L_{jki}}{2u_{jk}}\right)
+\partial_k\partial_i\left(\frac{L_{kij}}{2u_{ki}}\right)
-\partial_i\partial_j\partial_k\left(\frac{1}{2}\ln \frac{u_{ijk}-L_{ijk}}{u_{ijk}+L_{ijk}}\right)=0,
\end{equation}
$i<j<k$. Equations \eqref{5-c2} and \eqref{5-c3} are compatible as
$$\partial_l E_{ijk}=\partial_i\partial_j B_{ijlk}+\partial_j\partial_k B_{jkli}+\partial_k\partial_i B_{kilj}
+\partial_i\partial_j\partial_k C_{ijkl}.
$$
The skew-symmetry of $L_{ijk}$ implies that the expressions
$$
F_{ijk}=L_{ijk}+\frac{1}{2}u_{ijk}\ln \frac{u_{ijk}-L_{ijk}}{u_{ijk}+L_{ijk}}
$$
are also totally skew-symmetric, and can be combined into a 3-form
\begin{equation}\label{LMF}
\Omega=\sum_{i<j<k} F_{ijk}\, dx^i\wedge dx^j\wedge dx^k;
\end{equation}
note that
$$
\frac{\partial F_{ijk}}{\partial u_{ij}}=\frac{L_{ijk}}{2u_{ij}},\quad
\frac{\partial F_{ijk}}{\partial u_{jk}}=\frac{L_{ijk}}{2u_{jk}},\quad
\frac{\partial F_{ijk}}{\partial u_{ik}}=\frac{L_{ijk}}{2u_{ik}},\quad
\frac{\partial F_{ijk}}{\partial u_{ijk}}=\frac{1}{2}\ln \frac{u_{ijk}-L_{ijk}}{u_{ijk}+L_{ijk}}.
$$
It turns out that $\Omega$ plays the role of a Lagrangian multiform for the Darboux system expressed in terms of potential $u$, that is, all PDEs satisfied by $u$ result from $\Omega$ by calculating multiform Euler-Lagrange equations  \cite{LN2009, S2016, SV2016}. In our case,  the full system of multiform Euler-Lagrange equations  consists of four relations:
\begin{equation}\label{MF1}
\frac{\delta_{ijk} F_{ijk}}{\delta u}= E_{ijk}=0,
\end{equation}
\begin{equation}\label{MF2}
\frac{\delta_{ijk} F_{ijk}}{\delta u_k}-\frac{\delta_{ijl} F_{ijl}}{\delta u_l}= \partial_k^{-1}\partial_l^{-1}\left(\partial_kE_{ijl}-\partial_lE_{ijk}+\partial_i\partial_jB_{ijlk}\right)=0,
\end{equation}
\begin{equation}\label{MF3}
\frac{\delta_{ijk} F_{ijk}}{\delta u_{jk}}+\frac{\delta_{ikl} F_{ikl}}{\delta u_{kl}}
+\frac{\delta_{ilj} F_{ilj}}{\delta u_{lj}}= C_{jkli}-\partial_i\ln A_{ijkl}=0,
\end{equation}
\begin{equation}\label{MF4}
\frac{\delta_{ijk} F_{ijk}}{\delta u_{ijk}}-\frac{\delta_{jkl} F_{jkl}}{\delta u_{jkl}}
+\frac{\delta_{kli} F_{kli}}{\delta u_{kli}}-\frac{\delta_{lij} F_{lij}}{\delta u_{lij}}= \ln A_{ijkl}=0.
\end{equation}
Here, $\frac{\delta \cdot\cdot\cdot}{\delta \cdot\cdot\cdot}$ denote variational derivatives
which, for functions depending on the second-order and third-order derivatives of potential $u$  (such as $F_{ijk}$), take the form
$$
\frac{\delta_{ijk}}{\delta u}=
\partial_i\partial_j \frac{\partial}{\partial u_{ij}}
+\partial_j\partial_k \frac{\partial}{\partial u_{jk}}
+\partial_k\partial_i \frac{\partial}{\partial u_{ki}}
-\partial_i\partial_j\partial_k \frac{\partial}{\partial u_{ijk}},
$$
$$
\frac{\delta_{ijk}}{\delta u_k}=
-\partial_i \frac{\partial}{\partial u_{ik}}
-\partial_j \frac{\partial}{\partial u_{jk}}
+\partial_i\partial_j \frac{\partial}{\partial u_{ijk}},
$$
$$
\frac{\delta_{ijk}}{\delta u_{jk}}=
\frac{\partial}{\partial u_{jk}}
-\partial_i \frac{\partial}{\partial u_{ijk}},
$$
$$
\frac{\delta_{ijk}}{\delta u_{ijk}}=
\frac{\partial}{\partial u_{ijk}}.
$$
For every triple of distinct indices, equations (\ref{MF1}) give sixth-order PDEs (\ref{5-c3}) for the potential $u$. For every quadruple of distinct indices, equations (\ref{MF4}) give third-order PDEs (\ref{5-c1}). Finally, for every quadruple of distinct indices, equations (\ref{MF3}) give fourth-order PDEs (\ref{5-c2}). The remaining relations (\ref{MF2}) will be satisfied identically.

\medskip

The differential of $\Omega$  is a 4-form,
$$
d\Omega=\sum_{i<j<k<l} \omega_{ijkl}\, dx^i\wedge dx^j\wedge dx^k\wedge dx^l,
$$
where
$$
\omega_{ijkl}=\partial_{i}F_{jkl}-\partial_{j}F_{ikl}+\partial_{k}F_{ijl}-\partial_{l}F_{ijk}.
$$
Direct calculation gives
$$
\omega_{ijkl}=\widetilde{A}_{ijkl}-u_{ijkl}\ln A_{ijkl},
$$
which shows that $\Omega$ is closed modulo PDEs satisfied by $u$.

\bigskip

Let us return to the compatibility issue. Although the full set of PDEs for $u$  must be involutive due to involutivity of  Darboux system (\ref{c1Dar}),  this can also be verified by (rather nontrivial) direct calculation. For simplicity, consider the case $n=4$. Then the full set of PDEs for $u$  reduces to a single pair of fourth-order and  third-order PDEs (that imply all other equations for $u$, namely sixth-order ones,  as differential consequences).
The compatibility of a pair of PDEs can be verified by computing their Mayer bracket \cite{Krug}.

\section{Concluding remarks}

\begin{itemize}

\item It was demonstrated by Nijhoff in \cite{Nijhoff2023, Nijhoff2024} that the full hierarchy of the Darboux system (including its continuous, semi-discrete and fully discrete commuting flows), possesses a Lagrangian multiform formulation. It would be nice to add the corresponding `discrete part' to the  Lagrangian multiform (\ref{LMF}), expressed in terms of potential $u$.

\item Various explicit solutions of sixth-order PDE (\ref{c1PDE}) can be obtained by taking special metrics of diagonal curvature and calculating the corersponding potential
$u$. Thus, the general class of St\"ackel metrics in 3D (diagonal metrics  whose equations of geodesics possess two additional quadratic integrals) give rise to  potentials $u$ of the form
$$
u=-\ln\det \tau \quad {\rm where} \quad
\tau=\left( \begin{array}{ccc}
1 & 1& 1\\
\phi_1(x^1) & \phi_2(x^2) &\phi_3(x^3) \\
\psi_1(x^1) & \psi_2(x^2) &\psi_3(x^3)
\end{array}
\right);
$$
here $\phi_i(x^i)$ and $\psi_i(x^i)$ are arbitrary functions of  the indicated arguments. This can be verified by direct substitution into the equation.
\end{itemize}

\bigskip

{\noindent \bf Data availability statement.}
Data sharing is not applicable to this article as no datasets were generated or analysed during the current study.

\medskip
{\noindent \bf Conflict of interest statement.}
The corresponding author states that there is no conflict of interest.

\bigskip

\section*{Acknowledgments}

We thank Matteo Casati for numerous useful discussions. We also thank the reviewers of this paper for important suggestions. 
The research of MVP was partially supported by the  NSFC (Grant No.12431008).

\end{document}